\documentclass[10pt,journal]{IEEEtran}
\usepackage{cite,calc}
\usepackage[cmex10]{amsmath} 
\usepackage{amsfonts,amssymb,graphics,subfigure,epsfig,graphics}
\usepackage[mathcal]{euscript} 

\usepackage{mathrsfs}

\newcommand\nc\newcommand
\nc\bfa{{\boldsymbol a}}\nc\bfA{{\boldsymbol A}}\nc\cA{{\mathcal A}}
\nc\bfb{{\boldsymbol b}}\nc\bfB{{\boldsymbol B}}\nc\cB{{\mathcal B}}
\nc\bfc{{\boldsymbol c}}\nc\bfC{{\boldsymbol C}}\nc\cC{{\mathcal C}}
\nc\sC{{\mathscr C}}
\nc\bfd{{\boldsymbol d}}\nc\bfD{{\boldsymbol D}}\nc\cD{{\mathcal D}}
\nc\bfe{{\boldsymbol e}}\nc\bfE{{\boldsymbol E}}\nc\cE{{\mathcal E}}
\nc\bff{{\boldsymbol f}}\nc\bfF{{\boldsymbol F}}\nc\cF{{\mathcal F}}
\nc\bfg{{\boldsymbol g}}\nc\bfG{{\boldsymbol G}}\nc\cG{{\mathcal G}}
\nc\bfh{{\boldsymbol h}}\nc\bfH{{\boldsymbol H}}\nc\cH{{\mathcal H}}
\nc\bfi{{\boldsymbol i}}\nc\bfI{{\boldsymbol I}}\nc\cI{{\mathcal I}}
\nc\bfj{{\boldsymbol j}}\nc\bfJ{{\boldsymbol J}}\nc\cJ{{\mathcal J}}
\nc\bfk{{\boldsymbol k}}\nc\bfK{{\boldsymbol K}}\nc\cK{{\mathcal K}}
\nc\bfl{{\boldsymbol l}}\nc\bfL{{\boldsymbol L}}\nc\cL{{\mathcal L}}
\nc\bfm{{\boldsymbol m}}\nc\bfM{{\boldsymbol M}}\nc\sM{{\mathscr M}}
\nc\bfn{{\boldsymbol n}}\nc\bfN{{\boldsymbol N}}\nc\cN{{\mathcal N}}
\nc\bfo{{\boldsymbol o}}\nc\bfO{{\boldsymbol O}}\nc\cO{{\mathcal O}}
\nc\bfp{{\boldsymbol p}}\nc\bfP{{\boldsymbol P}}\nc\cP{{\mathcal P}}
\nc\bfq{{\boldsymbol q}}\nc\bfQ{{\boldsymbol Q}}\nc\cQ{{\mathcal Q}}
\nc\bfr{{\boldsymbol r}}\nc\bfR{{\boldsymbol R}}\nc\cR{{\mathcal R}}
\nc\bfs{{\boldsymbol s}}\nc\bfS{{\boldsymbol S}}\nc\cS{{\mathcal S}}
\nc\bft{{\boldsymbol t}}\nc\bfT{{\boldsymbol T}}\nc\cT{{\mathcal T}}
\nc\bfu{{\boldsymbol u}}\nc\bfU{{\boldsymbol U}}\nc\cU{{\mathcal U}}
\nc\bfv{{\boldsymbol v}}\nc\bfV{{\boldsymbol V}}\nc\cV{{\mathcal V}}
\nc\bfw{{\boldsymbol w}}\nc\bfW{{\boldsymbol W}}\nc\cW{{\mathcal W}}
\nc\bfx{{\boldsymbol x}}\nc\bfX{{\boldsymbol X}}\nc\cX{{\mathcal X}}
\nc\bfy{{\boldsymbol y}}\nc\bfY{{\boldsymbol Y}}\nc\cY{{\mathcal Y}}
\nc\bfz{{\boldsymbol z}}\nc\bfZ{{\boldsymbol Z}}\nc\cZ{{\mathcal Z}}

\DeclareMathOperator{\supp}{supp}
\DeclareMathOperator{\rank}{rank}

\DeclareMathOperator{\wt}{wt}

\newcommand{\h}{h_\mathrm{B}}
\newcommand{\avg}{{\mathbb E}}
\newcommand{\dist}{d_\mathrm{H}}

\newtheorem{theorem}{Theorem}
\newtheorem{definition}{Definition}
\newtheorem{lemma}[theorem]{Lemma}
\newtheorem{proposition}[theorem]{Proposition}
\newtheorem{corollary}[theorem]{Corollary}

\newtheorem{remark}{\indent Remark}

\newcommand\reals{{\mathbb R}}

\newcommand\ff{{\mathbb F}}

\begin{document}

\title{Update-Efficiency and Local Repairability Limits\\
       for Capacity Approaching Codes}

\author{Arya Mazumdar~\IEEEmembership{Member,~IEEE}, Venkat Chandar, and Gregory W. Wornell~\IEEEmembership{Fellow,~IEEE}

\thanks{Manuscript received May 1, 2013, revised Oct.\ 1, 2013.  This
  work was supported in part by the US Air Force Office of Scientific
  Research under Grant No.~FA9550-11-1-0183, and by the National
  Science Foundation under Grant No.~CCF-1017772, and a grant from
  University of Minnesota.}  \thanks{A.~Mazumdar was with the Research
  Laboratory of Electronics, Massachusetts Institute of Technology,
  Cambridge, MA~~02139.  He is now with the Department of Electrical
  and Computer Engineering, University of Minnesota, Minneapolis,
  MN~~55455 (Email: arya@umn.edu).}  \thanks{V.~Chandar is with MIT
  Lincoln Laboratory, Lexington, MA~~02420 (Email: vchandar@mit.edu).}
\thanks{G.~W.~Wornell is with the Department of Electrical Engineering
  and Computer Science, Massachusetts Institute of Technology,
  Cambridge, MA~~02139 (Email: gww@mit.edu).}  \thanks{The results of
  this paper were presented in part at the 2012 IEEE International
  Symposium on Information Theory, Cambridge, MA, July 1-6, 2012, and
  at the Information Theory and Applications Workshop, San Diego, CA,
  Feb.\ 9-14, 2013.}}

\maketitle

\begin{abstract}
Motivated by distributed storage applications, we investigate the
degree to which capacity achieving encodings can be efficiently
updated when a single information bit changes, and the degree to which
such encodings can be efficiently (i.e., locally) repaired when single
encoded bit is lost.

Specifically, we first develop conditions under which optimum
error-correction and update-efficiency are possible, and establish
that the number of encoded bits that must change in response to a
change in a single information bit must scale logarithmically in the
block-length of the code if we are to achieve any nontrivial rate with
vanishing probability of error over the binary erasure or binary
symmetric channels.  Moreover, we show there exist capacity-achieving
codes with this scaling.

With respect to local repairability, we develop tight upper and lower
bounds on the number of remaining encoded bits that are needed to
recover a single lost bit of the encoding.  In particular, we show
that if the code-rate is $\epsilon$ less than the capacity, then for
optimal codes, the maximum number of codeword symbols required to
recover one lost symbol must scale as $\log1/\epsilon$.

Several variations on---and extensions of---these results are also
developed.
\end{abstract}

\begin{IEEEkeywords}
error-control coding, linear codes, locally updatable codes, locally
recoverable codes, low-density generator matrix codes, low-density
parity check codes
\end{IEEEkeywords}

\section{Introduction}

There is a growing need to provide reliable distributed data storage
infrastructure in highly dynamic and unreliable environments.  Storage
nodes (servers) can switch between on- and off-line frequently, due to
equipment failures, breaks in connectivity, and maintenance activity.
Corruptions of the data can also arise.  Furthermore, the data itself
is often changing frequently, requiring the constant updating of
storage nodes.  

These characteristics place additional demands on the error-control
coding typically considered for such systems, and there are important
questions about the degree to which those demands can be accommodated,
which the community has begun to investigate in recent years.

In this paper, we focus on two specific such demands.  The first is
that the coding be \emph{update-efficient} or
\emph{locally-updatable}, i.e., that small changes in the data require
only small changes in its coded representation.  The second is that
the coding be \emph{locally-repairable} or \emph{recovery-efficient},
i.e., that small portions of the coded representation that are lost
can be recovered from correspondingly small portions of the rest of
the encoding.  The degree to which these demands can be accommodated
affects the bandwidth and energy consumption requirements of such
infrastructure.

\subsection{Update-Efficiency}
\label{sec:update}

While the notion of update-efficiency is implicit in work on array
codes---see, for example, \cite{xu1999x,jin2009p}---the first
substantial analysis of update-efficiency appears in
\cite{anthapadmanabhan2010update}.  

When changes in the data are significant, then updating a linear
fraction of the encoding is unavoidable for capacity-approaching
codes.  However, when the changes are incremental, updating a
sublinear fraction of the encoding can be sufficient.  For example,
\cite{anthapadmanabhan2010update} considers codes over a binary
alphabet and shows the existence of a code that achieves the capacity
of the binary erasure channel (BEC) with the property that any
single-bit change in the message requires only a logarithmic (in the
blocklength) number of bits to be updated in the codeword.

Among related work, \cite{jule2011some} examines the update-efficiency
of codes with structure---specifically, random linear codes---and
\cite{rawat2011update}, in the context of distributed storage, shows
that using the randomized codes proposed in
\cite{anthapadmanabhan2010update} it is possible to have
capacity-achieving update-efficient codes for the BEC that also
minimize a particular measure of the bandwidth required to replace
coded data lost as the result of a node failure.\footnote{This notion
  of repair bandwidth turns out to be somewhat different than the
  notion of local repairability we describe in Section~\ref{sec:lr}.}

In this paper, we begin with a brief discussion of update-efficient
codes that can correct arbitrary sets of (adversarial)
errors/erasures, since such adversarial models are common in current
storage applications.  In this case, update-efficiency and
error-correctability are directly conflicting properties.  In
particular, it is impossible to correct more than $\sim t/2$ errors
(or $\sim t$ erasures) with a code that needs at most $t$ bits of
update for any single-bit change in the message.  This is because the
minimum pairwise distance between the codewords (i.e., the
\emph{minimum distance} of the code) is upper bounded by $t$.  We
discuss several properties of linear codes that are useful for
constructing good update-efficient adversarial-error-correcting codes,
i.e., codes that achieve the best possible tradeoff.  Perhaps the most
interesting observation for this scenario is that if there exists a
linear code with a given rate and minimum distance, then there exists
another linear code with same parameters that is as update-efficient
as possible.

For the remainder of our development on update-efficiency, we then
turn our attention to the random failure model, where much better
\emph{average} performance is possible.  We begin with a simple
derivation of one of the main propositions of
\cite{anthapadmanabhan2010update}, i.e., that there exist linear codes
that achieve the capacity of the BEC such that for any single-bit
change in the message, only $O(\log n)$ bits have to be updated in a
codeword of length $n$.  However, our main result is the converse
statement: we show that if a linear code of positive rate achieves
arbitrarily small probability of error over a BEC, then a single-bit
change in the message must incur an updating of $\Omega(\log n)$ bits
in the codeword.  In addition, we estimate $\gamma >0$ such that there
cannot exist a linear code with positive rate and arbitrarily small
probability of error that requires fewer than $\gamma \log n$ bit
updates in the codeword per single bit change in the message.

\subsection{Local Repairability}
\label{sec:lr}

The potential for a code to be locally-repairable---which is, in some
sense, a notion dual to that of update-efficiency---was first analyzed
in \cite{gopalan2012locality}. 

For significant losses of codeword symbols, the minimum distance
properties of the code naturally characterize the worst-case
requirements to repair the encoding.  However, the ability to repair
smaller losses depends on other properties of the code.  For example,
\cite{gopalan2012locality} shows that it is possible to recover any
single symbol of any codeword from at most a constant number of other
symbols of the codeword, i.e., a number of symbols that does not grow
with the length of the code.    

From this perspective, in practice one might hope to have codes with
both a large minimum distance and that require the fewest number of
other symbols to recover single symbol losses.

To this end, \cite{gopalan2012locality,papailiopoulos2012locally}
consider locally repairable codes that also correct a prescribed
number of adversarial errors (or erasures), and develop a trade-off
between the local repairability and error-correctability.  In
particular, it is shown that for a $q$-ary linear code ($q\ge2$) of
blocklength $n$, the minimum distance $d$ satisfies
\begin{equation*} 
d \le n -k -\left\lceil\frac{k}{r}\right\rceil +2,
\end{equation*}
where $k$ is the code dimension and $r$ 
symbols required to reconstruct a single symbol of the code (referred
to as the local repairability)..

Such results can be generalized to nonlinear codes of any desired
alphabet size.  Indeed, \cite{cadambe2013upper} shows that for any
$q$-ary code with size $M$, local repairability $r$, and minimum
distance $d$, we have
\begin{equation}
\log M \le \min_{1\le t \le \left\lceil \frac{n}{r+1} \right\rceil} \Big[ tr + \log A_q(n  - t(r+1),d) \Big],
\end{equation}
where $A_q(n,d)$ is the maximum size of a $q$-ary code of length $n$
and distance $d$.  

It should be noted that, in contrast to the case of update-efficiency,
which requires a number of codewords to be close to each other, there
is no immediate reason that a code cannot have both good
local repairability and good error-correctability.  Perhaps not
surprisingly, there has, in turn, been a growing literature exploring
locally repairable codes with other additional properties; see, e.g.,
\cite{kamath2013explicit,rawat2012optimal,tamo2013optimal}.

In this paper, we focus on probabilistic channel models that take into
account the statistics of node failure, and optimize average
performance.  To the best of our knowledge the associated capacity
results for locally repairable codes have yet to be developed.
Indeed, the analysis of local repairability in the existing literature
is almost invariably restricted to an adversarial model for node
failure.  While combinatorially convenient, there is no guarantee the
resulting codes are good in an average sense.

In our development, we first show that it is possible to construct
codes operating at a rate within $\epsilon$ of the capacity of the BEC
that have both local repairability $O(\log 1/\epsilon)$ and an
update-efficiency scaling logarithmically with the block-length.
However, our main result in this part of the paper is a converse
result establishing that the scaling $O(\log 1/\epsilon)$ is optimal
for a BEC---specifically, we establish that if the rate of a code that
achieves arbitrarily small probability of error over a BEC is
$\epsilon$ below capacity, then the local repairability is
$\Omega(\log 1/\epsilon)$.

\subsection{Channels with Errors}

Most of our development focuses on the case of ``hard'' node failures
that result data loss.  However, in some scenarios node failures are
``softer'', resulting in data corruption.  While the BEC is a natural
model for the former, it is the binary symmetric channel (BSC) that is
the corresponding model for the latter.  Much, but not all, of our
development carries over to the case of the BSC.

In particular, our results on the existence of capacity-achieving
codes that are both update-efficient and locally repairable also hold
for the BSC.  Likewise, our converse result for update-efficient
linear codes also holds the BSC.  However, we have an additional
converse result for general codes that applies only to the BSC, and
our converse result for local repairability applies only to the BEC.

\subsection{Organization} 

The organization of the paper is as follows.
Section~\ref{sec:notations} establishes notation that will be used
throughout the paper.  In Section~\ref{adversarial}, we discuss the
worst-case error-correction capability of an update-efficient code.
In Section~\ref{explicit} we show that there exist linear codes of
length $n$ and rate $\epsilon$ less than capacity, with
update-efficiency logarithmic in blocklength and local repairability
$O(\log1/\epsilon)$. In Section~\ref{sec:LDGM}, our main impossibility
results for capacity-achieving update-efficient codes are presented.
Subsequently, in Section~\ref{sec:recover} we address the local
repairability of capacity-achieving codes and deduce a converse result
that matches the achievability part.  In Section~\ref{sec:general}, we
give a generalized definition of update efficient codes.  In
Section~\ref{ratedistortion} we note that the notions of
update-efficiency and local repairability are also applicable to
\emph{source coding}, i.e., data compression, and we briefly discuss
the associated dual problem of lossy source coding in the context of
update-efficiency and local recoverability.  Finally,
Section~\ref{sec:conc} contains some concluding remarks.


\section{Definitions and Notation}\label{sec:notations}

First, we use BEC($p$) to denote a binary erasure channel with loss
probability $p$, which has capacity $1-p$.  Analogously, we use
BSC($p$) to denote a binary symmetric channel with crossover
probability $p$, which has capacity $1-\h(p)$ where $\h(p) =
-p\log_2(p) - (1-p)\log_2(1-p)$ is the binary entropy function.

Next, for our purposes a \emph{code} $\cC\in \ff_2^n$ is a collection
of binary $n$-vectors.  The \emph{support} of a vector $\bfx$ (written
as $\supp(\bfx)$) is the set of coordinates where $\bfx$ has nonzero
values. By the \emph{weight} of a vector we mean the size of support
of the vector. It is denoted as $\wt(\cdot)$.  All logarithms are
base-$2$ unless otherwise indicated.

Let $\sM$ be the set of all possible messages.  Usually accompanied
with the definition of the code, is an injective encoding map $\phi
:\sM \to \cC,$ which defines how the messages are mapped to
codewords. In the following discussion, let us assume $\sM = \ff^k_2$.
In an update-efficient code, for all $\bfx \in \ff^k_2,$ and for all
$\bfe\in \ff_2^k : \wt(\bfe) \le u,$ we have $\phi(\bfx +\bfe) =
\phi(\bfx) + \bfe',$ for some $\bfe' \in \ff_2^n: \wt(\bfe') \le t.$ A
special case of this is captured in the following definition.
\begin{definition}\label{defn:update}
The \emph{update-efficiency} of a code $\cC$ and the encoding $\phi$, is the maximum number of bits that needs to be changed
in a codeword when a single bit in the message is changed. A code $(\cC,\phi)$ has update-efficiency
$t$ if
for all $\bfx \in \ff^k_2,$ and for
all $\bfe\in \ff_2^k : \wt(\bfe) =1,$ we have $\phi(\bfx +\bfe) = \phi(\bfx) + \bfe',$
for some $\bfe' \in \ff_2^n: \wt(\bfe') \le t.$
\end{definition}

A linear code  $\cC\in \ff_2^n$ of dimension $k$ is a $k$-dimensional
subspace of the vector space $\ff_2^n$. For linear codes the
 mapping $\phi:\ff_2^k \to \cC$ is naturally given by 
a {\em generator matrix}: $\phi(\bfx) = \bfx^T G$,  for any $\bfx \in \ff_2^k$.
 There can be a number of generator matrices for a code
$\cC$, which correspond to different labelings of the codewords. 
By an $[n,k,d]$ code  we mean a linear code with length $n$, dimension
$k$ and minimum (pairwise) distance between the codewords $d$.
Linear codes form
the most studied classes of error-correcting codes, and have a number of
benefits in terms of representation and encoding and decoding complexity. 

 
For a linear code, when changing one bit in the message, the maximum number of bits that
need to be changed in the codeword is the maximum over the weights of the rows of the
generator matrix. Hence, for an update-efficient code,
we need a representation of the linear code 
where the maximum weight of the rows of the generator matrix is low.
\begin{proposition}\label{prop:linear}
A linear code $\cC$ will have update-efficiency $t$ if and only if
there is a generator matrix $G$ of $\cC$ with maximum row weight $t.$
\end{proposition}
\begin{IEEEproof}
It is easy to see that if the maximum number of ones in any row is bounded above by $t$, then
at most $t$ bits need to be changed to update one bit change in the message.

On the other hand if the code has update-efficiency $t$ then there must exist a labeling
$\phi$ that gives a sparse generator matrix. Specifically, each of the vectors $(1, 0,\dots,0), (0,1,\dots, 0),
\dots, (0,0\dots,1) \in \ff_2^k$ must produce codeword of weight at most $t$ under 
 $\phi.$ Therefore, the generator matrix given by $\phi$ will have
row weight at most $t.$
\end{IEEEproof}
This implies, given a linear code, to see whether it is update-efficient or not, we need to find
the sparsest basis for the code. A linear code with a sparse
basis is informally called a {\em low density generator matrix (LDGM)} code. 

\vspace{0.1in}

There are a number of different ways that local recovery could
be defined. The simplest
is perhaps the one given below, which insists that for each codeword 
symbol, there is a set of
at most $r$ codeword positions that need to be queried to recover the given
symbol with certainty. A weaker definition could allow adaptive queries,
i.e., the choice of which $r$ positions to query could depend on the
values of previously queried symbols. Finally, one could ask that instead
of obtaining the value of the codeword symbol with certainty, one
obtains the value with some probability significantly higher than $.5$ (probabilistic recovery).
We sketch all the arguments in this paper for the simplest definition, i.e., Defn.~\ref{defn:recov}.
 The
results can  be extended to probabilistic recovery without much change in the argument. 

\begin{definition}\label{defn:recov}
A code $\cC\subset \ff_2^n$ has {\em local recoverability} $r$, if for any $\bfx = (x_1, \ldots , x_n) \in
\cC$ and for any $1\le i\le n$, there exists a function $f_i:\ff_2^r \to \ff_2$ and 
indices $1\le i_1, \ldots , i_r\le n, i_j \ne i, 1\le j\le r,$
 such that $x_i = f_i(x_{i_1}, \ldots, x_{i_r})$.   
\end{definition}

A generator matrix $H$ of the null-space of a linear code $\cC$ is called a
parity-check matrix for $\cC$. It is to be noted that for any $\bfx \in \cC, H\bfx =0$. 
A {\em low density parity-check} (LDPC) code is a linear code
with a parity check matrix such that each row of the parity check matrix has a small (constant)
number of nonzero values.
The following proposition is immediate.
\begin{proposition}
If the maximum row-weight of a parity-check matrix of a code is $r$, then
the code has local recoverability at most $r$.
\end{proposition}

Hence, LDPC codes are locally recoverable. We note that this is almost
a necessary condition for local recoverability. To make this a necessary condition,
one needs to consider not just rows of the parity check matrix, but the entire
dual code. A necessary condition for local recoverability is that the 
dual code contains low weight codewords whose supports cover any of the $n$ coordinates.

\section{Adversarial channels}\label{adversarial}
The adversarial 
error model is ubiquitously studied in the storage literature.
In an adversarial error model, the channel is allowed to
introduce up to $s$ errors (or $2s$ erasures), and the location of these errors
can be chosen by an adversary. It is known that to correct $s$ 
errors ($2s$ erasures), the minimum distance of the code needs to be at least
$2s+1.$ However, if a code has update-efficiency $t$, then there must exist
two (in fact, many more) codewords that are within distance $t$ of each other.
Hence, small update-efficiency implies poor adversarial error correction capability,
and we cannot hope to find good codes if we adopt an adversarial error model. 
Nevertheless, before moving on to the main results of the paper, we briefly digress to
make some observations of theoretical interest regarding adversarial error models. 



In a code with minimum pairwise distance between codewords $d$,
the update-efficiency has to be at least $d$, because the nearest codeword is
at least distance $d$ away. That is, if the update-efficiency of the code $\cC$ is denoted by 
$t(\cC)$, then 
$$
t(\cC) \ge d(\cC),
$$ where $d(\cC)$ is the minimum distance of the code.  The main observation
we would like to make in 
this section is that the above bound is in fact achievable with
the best possible parameters of a linear code.  We have seen in
Section~\ref{sec:notations} that for a linear code $\cC$, the
update-efficiency is simply the weight of the maximum weight row of a
generator matrix. The following theorem is from \cite{dodunekov1985improvement}.

\begin{theorem}
Any binary linear code of length $n$, dimension $k$ and distance $d$ has a generator 
matrix consisting of rows of weight $\le d+s$, where
\begin{equation*} 
s = \Big(n - \sum_{j=0}^{k-1}\Big\lceil\frac{d}{2^j}\Big\rceil\Big) 
\end{equation*}
 is a nonnegative integer.
\end{theorem}
The fact that $s$ is a non-negative integer also follows from the
well-known Griesmer bound \cite{MS1977}, which states that for any
linear code with dimension $k$, distance $d$, and length $n\ge
\sum_{j=0}^{k-1}\lceil d/2^j \rceil$.

\begin{corollary}
For any linear $[n,k,d]$ code $\cC$ with update-efficiency $t$,
\begin{equation*} 
d \le t \le d +  \Big(n - \sum_{j=0}^{k-1}\Big\lceil\frac{d}{2^j}\Big\rceil\Big).
\end{equation*}
\end{corollary}
It is clear that for codes achieving the Griesmer bound with
equality, the update-efficiency is exactly equal to the minimum distance,
i.e., the best possible.
There are a number of families of codes that achieve
the Griesmer bound. For examples of such families and their 
characterizations, we refer the reader to \cite{belov1974construction,helleseth1983new}.

\vspace{0.1in}
{\em Example:}
 Suppose $\cC$ is a $[n =2^m -1,k  = 2^m -1-m, 3]$ Hamming code. 
 For this code 
 $$
 t(\cC) \le 3 + (n - 3 -2 -(k-2)) = n-k = m = \log(n+1).
 $$ 
 One can easily achieve update-complexity $1+ \log(n+1)$ for Hamming codes.
 Simply bring any $k\times n$ generator matrix of Hamming code into
 systematic form, 
 resulting in the maximum weight of a row being bounded above by 
 $1+ (n-k) = 1+\log(n+1)$. This special case was mentioned in \cite{anthapadmanabhan2010update}. 
 This can also be argued from the point of view that as the
 generator polynomial of a Hamming code (cyclic code) has degree $m$,
 the maximum row-weight of a generator of a Hamming code will be at most
 $m+1 = \log(n+1)+1.$ 
 
 However, we can do even better by 
  explicitly constructing a generator matrix for the Hamming code in the following way.
  Let us index the columns of the generator matrix by $1, 2, \dots, 2^m-1$, and use
the notation $(i,j,k)$ to denote the vector with exactly three $1$'s, located at
positions $i, j,$ and $k$. Then, the Hamming code has a generator matrix given by
the row vectors
$(i, 2^j, i+2^j)$ for $1\le j \le m-1, 1\le i < 2^j.$
%
 This shows that for all $n$, Hamming codes have update-efficiency
only $3.$ To prove this without explicitly constructing a generator matrix, 
and to derive some other consequences, we need 
the following theorem by Simonis \cite{simonis1992generator}.

\begin{theorem}
Any $[n,k,d]$ binary linear code can be transformed into a code with the
same parameters that has a generator
matrix consisting of only weight $d$ rows.
\end{theorem}
The implication of this theorem is the following: if there exists an $[n,k,d]$ linear code,
then there exists an $[n,k,d]$ linear code with update-efficiency $d.$ 
The proof of \cite{simonis1992generator} can be presented as an algorithm that transforms
any linear code, given its parameters $[n,k,d]$ and a generator matrix, into an update-efficient 
linear code (a code with update-efficiency equal to
the minimum distance).  The algorithm, in time possibly exponential in $n$, produces a 
new generator matrix with all rows having weight $d$. 
It is of interest to find a polynomial time (approximation) algorithm for the procedure, that is,
a generator matrix with all rows having weight within $d(1+\epsilon)$ for some small $\epsilon$.

On the other hand, the above theorem says that there exists a linear $[n=2^m-1, k+2^m-1-m,3]$
code that has update-efficiency only $3.$ All codes with these parameters are
equivalent to the Hamming code with the same parameters up to a permutation of coordinates \cite{huffman2003fundamentals},
providing an alternate proof that Hamming codes have update-efficiency $3$.

Analysis of update-efficiency for BCH codes and other linear codes 
is of independent interest. In general, finding a sparse basis for a
linear code given its generator matrix seems to be a hard problem, although
the actual complexity class of the problem merits further investigation.
Recently, a sparse basis is presented for $2$-error-correcting  BCH codes in \cite{grigorescu2012explicit}.

We emphasize that the previous remarks, although of theoretical interest, are
fairly strong negative results suggesting
that update efficiency and error correction are fundamentally incompatible requirements.
Although this is the case for adversarial models, if we allow randomization so that
the code can be chosen in a manner unknown to the adversary, 
then it is possible to fool the adversary. In fact, with a randomized code
it is possible to correct $~pn$ adversarial errors with a code rate close to the
capacity of BSC($p$) using ideas put forward in \cite{langberg2004private}. The randomized
code can be chosen to simultaneously have good update efficiency, as shown in the case of erasure
channels by \cite{anthapadmanabhan2010update}. In the rest of the paper, instead of choosing a code
at random to fool an adversary, we consider the classical information theoretic scenario of a random,
rather than an adversarial, channel, although we note that our results easily extend to the case of
using randomization to defeat an adversary.




\section{Existence of good codes}\label{explicit}
In this section, our aim is to show, in a rather simple way, that
there exist linear codes of length $n$ that
\begin{enumerate}
\item  have rate $\epsilon$ less than capacity, $\epsilon >0$, 
\item achieve arbitrarily small probability of error,
\item  have update-efficiency $O(\log n)$ and 
\item have local recoverability $O(\log 1/\epsilon)$.
\end{enumerate}

It is relatively easy to construct a code with 
local recoverability $O(\log 1/\epsilon)$ 
that achieves
 capacity over the BSC or BEC with  an $\epsilon$ gap. One can in principle choose the rows of the parity-check
  matrix randomly
from all low weight vectors, and argue that this random ensemble contains many 
codes that achieve the capacity of the binary symmetric channel (BSC) up to an
additive term $\epsilon$. 
Indeed, LDPC codes achieve the capacity of the binary symmetric channel \cite{gallager1962low}.

Similarly, one may try to construct a low row-weight generator matrix
randomly to show that the ensemble average performance achieves
capacity.  In this direction, some steps have been taken in
\cite{kakhakicapacity}. However, these constructions fail to achieve
local recoverability and update-efficiency simultaneously.  Also, in
\cite{asteris2012repairable}, parallel to a part of our work
\cite{mazumdar2012update}, a rateless code construction was proposed
that achieves both $O(\log k)$ update-efficiency and
local repairability, $k$ being the dimension of the code.  Below, we
describe one simple and intuitive construction that simultaneously
achieves $O(\log k)$ update efficiency and $O(\log 1/\epsilon)$
local repairability, where $\epsilon$ is the gap to capacity.


It is known that for for every $\epsilon >0$ and any sufficiently large $m$, there exists a linear code of
length $m$ and rate $1-\h(p)-\epsilon$ that has probability of incorrect decoding
at most $2^{-E(p,\epsilon)m}.$ There are numerous evaluations of this result and 
estimates of $E(p,\epsilon) >0.$ We refer the reader to \cite{barg2002random} as an example.
Below, $m,n,mR, nR, n/m$ are assumed to be integers. Floor and ceiling functions should be used
where appropriate. However, we avoid using them to maintain clarity in the argument, unless the meaning
 is not obvious from the context.  

Let $m = \frac{(1+\alpha)\log n }{E(p,\epsilon)}$, 
$\epsilon,\alpha >0$.  We know that
for sufficiently large $n$, there exists a linear code $\hat{\cC}$
given by the $mR \times m$ generator matrix $\hat{G}$ with rate $R =
1-\h(p)-\epsilon$ that has probability of incorrect decoding at most
$2^{-E(p,\epsilon)m}.$

Let $G$ be the $nR \times n$ matrix that is the Kronecker product of $\hat{G}$ and the 
$n/m \times n/m$ identity matrix $I_{n/m}$, i.e., 
$$
G = I_{n/m} \otimes \hat{G}.
$$
Clearly a codeword of the code $\cC$  given by $G$ is given by $n/m$ codewords of
the code $\hat{\cC}$ concatenated side-by-side. The probability of error 
of $\cC$ is therefore, by the union bound, at most
\begin{equation*} 
\frac{n}{m}2^{-E(p,\epsilon)m} = \frac{n E(p,\epsilon)}{(1+\alpha) n^{1+\alpha}\log n} = \frac{E(p,\epsilon)}{(1+\alpha) n^{\alpha}\log n}. 
\end{equation*}

However, notice that the generator matrix has row weight bounded above
by $m = (1+\alpha)/E(p,\epsilon)\log n$.  Hence, we have constructed a
code with update-efficiency $(1+\alpha)/E(p,\epsilon)\log n,$ and rate
$1-\h(p)-\epsilon$ that achieves a probability of error less than
$E(p,\epsilon)/[(1+\alpha) n^{\alpha}\log n]$ over a BSC($p$).

%
%


We modify the above construction slightly to produce codes that also possess
good local recoverability.
It is known that LDPC codes achieve a positive error-exponent. That is, 
   for every $\epsilon >0$ and any sufficiently large $m$, there exist an LDPC code of
length $m$ and rate $1-\h(p)-\epsilon$ that has check degree (number of $1$s in a
row of the parity-check matrix) 
  at most $O(\log 1/\epsilon)$, and
 probability of incorrect decoding
at most $2^{-E_L(p,\epsilon)m}$, for some $E_L(p,\epsilon)>0$ \footnote{
There are several works, such as  \cite{gallager1962low, miller2001bounds}, that discuss this result.
For example, we refer the reader to see Thm.~7 and Eqn.~(17) of \cite{miller2001bounds} for a derivation of
the fact.}.
This code
will be chosen as $\hat{\cC}$ in the above construction, and $\hat{G}$ can be any generator 
matrix for $\hat{\cC}$.

The construction now follows without any more changes. We have, $m =
(1+\alpha)/E_L(p,\epsilon) \log n,$ an integer, $\epsilon,\alpha >0$,
and $G = I_{n/m} \otimes \hat{G}$.

%
%

Now, the generator matrix has row weight bounded above by $m =
(1+\alpha)/E_L(p,\epsilon) \log n.$ So, the code has update-efficiency
$(1+\alpha)/E_L(p,\epsilon) \log n,$ rate $1-\h(p)-\epsilon,$ and
achieves probability of error less than
$E_L(p,\epsilon)/[(1+\alpha) n^{\alpha}\log n]$ over a BSC($p$).

Moreover, the parity-check matrix of the resulting code will be block-diagonal,
with each block being the parity-check matrix of the code $\hat{\cC}$.
The parity-check matrix of the overall code has row-weight $O(\log 1/\epsilon)$.
Hence, any codeword symbol can be recovered from at most $O(\log 1/\epsilon)$
other symbols by solving one linear equation. Therefore, we have the following result.

\begin{theorem}\label{thm:exists}
There exists a family of linear codes $\cC_n$ of length $n$ and rate $1-\h(p)-\epsilon$, that have probability of 
error over BSC($p$) going to $0$ as $n \to \infty$. These codes simultaneously achieve update-efficiency $O(\log n/E_L(p,\epsilon))$ and
local recoverability $O(\log 1/\epsilon)$.
\end{theorem} 

Hence, it is possible to simultaneously achieve local recovery and update-efficiency
with a capacity-achieving code on BSC($p$). 
But note that, this came with a price: namely, the decay of probability of error has become only 
polynomial as opposed to being exponential. 
A similar result is immediate  for BEC($p$).

\section{Impossibility results for update-efficiency}\label{sec:LDGM}

In this section, we show that for suitably small $\gamma$, no code can simultaneously achieve
capacity  and have update-efficiency
 better than $\gamma \log n$, $n$ blocklength. 
 More precisely, we give  the following converse results. 
 \begin{enumerate}
 \item {\em Linear codes.} Linear codes of positive 
 rate cannot have arbitrarily small probability of error 
 and update-efficiency better than $\gamma_1(p) \log n, \gamma_1(p) >0$  when used over the BEC (Thm.~\ref{thm:converse2}).
 Since a BSC is degraded with respect
to a BEC, this  result implies same claim for BSC as well.
To see that BSC($p$) is a degraded version of a BEC with erasure probability
$2p$, one can just concatenate BEC($2p$) with a channel with ternary input $\{0,1,?\}$ 
and binary output $\{0,1\}$, such that with probability $1$ the inputs $\{0,1\}$ remain the same, and
with uniform probability $?$ goes to $\{0,1\}$.  

 \item {\em General codes.}
Any (possibly non-linear) code with positive rate cannot have  update-efficiency
better than $\gamma_2(p) \log n, \gamma_2(p) >0$, and vanishing probability of error
when transmitted over BSC.  The value of $\gamma_2(p)$ that we obtain in this
case is  larger than $\gamma_1(p)$ of linear codes; moreover this  result applies to more
general codes than the previous (Thm.~\ref{thm:converse3}). But we have not been able to extend it to the BEC.
It could be interesting to explore whether nonlinear codes of
positive rate must have at least logarithmic update
efficiency for the BEC.

\item {\em LDGM ensemble.} We also show that
for the ensemble of LDGM codes with fixed row-weight $\gamma_3(p) \log n, \gamma_3(p)>0,$
 almost all codes have probability of error $\sim1$ when transmitted over a BSC (Thm.~\ref{thm:converse}).
 The value of $\gamma_3(p)$ in this case is much larger than the previous two cases.
\end{enumerate}

A plot providing the lower bound on update-efficiency of ``good''
codes is presented in Fig.~\ref{fig:converse}. In this figure, the
values of $\alpha$, the constant multiplier of $\ln n$, as a function
of BSC flip probability $p$ is plotted.  The plot contains results of
Theorems~\ref{thm:converse2}, \ref{thm:converse3} and
\ref{thm:converse}.  Note that $\gamma_1(p),\gamma_3(p) \to \infty$ as $p\to 1/2$
for general codes (Theorem~\ref{thm:converse3}) and the LDGM ensemble
(Theorem~\ref{thm:converse}).

\begin{figure}[tbp]
\centerline{\includegraphics[width=3.75in]{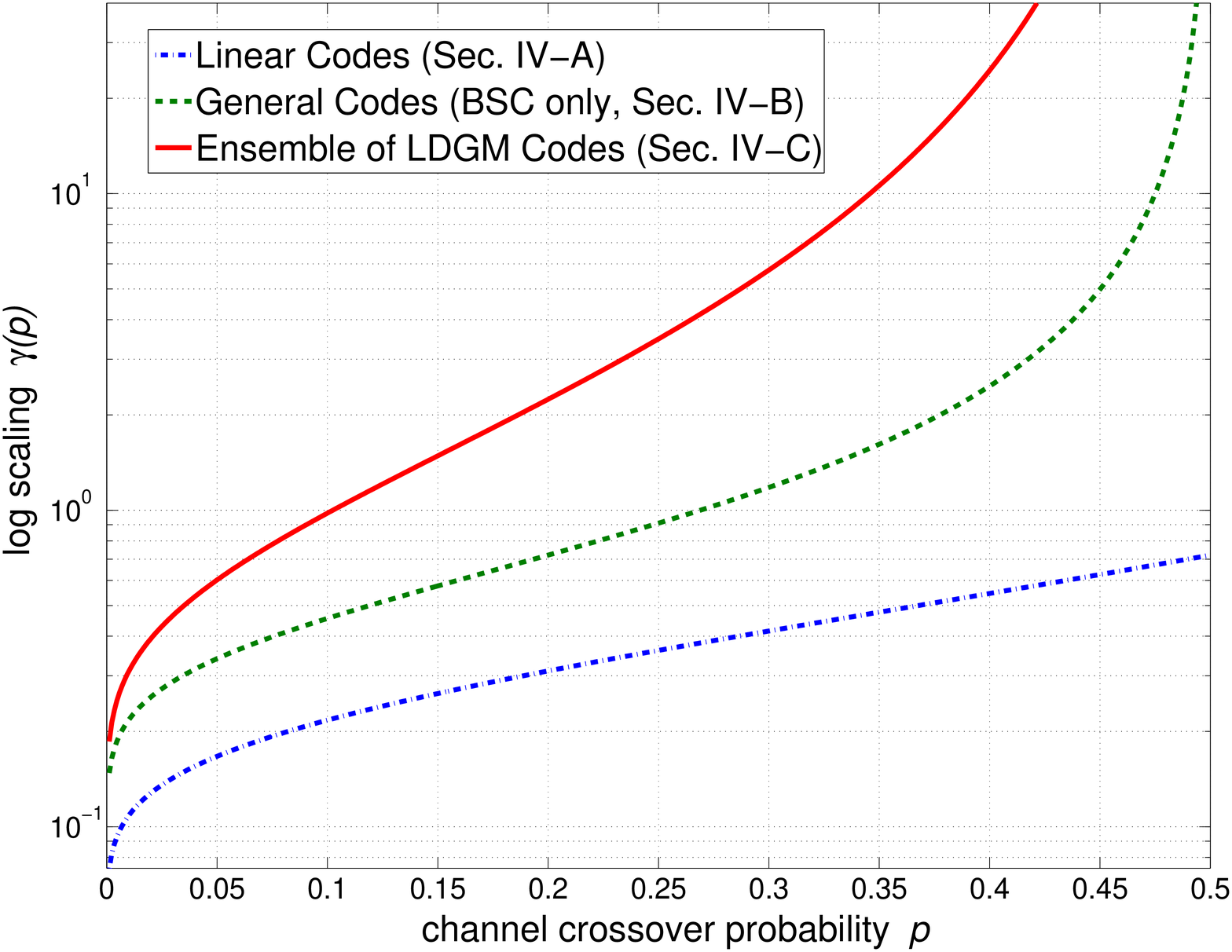}}
\caption{The plot of constant factors of $\ln n$ from Theorems \ref{thm:converse2}, \ref{thm:converse3} and \ref{thm:converse}. The fact that 
  the bound for linear codes appears below the bound for general codes is a artifact of the bounding technique: the bound for general codes
  is not extendable to BEC, but the bound for linear codes
  is. \label{fig:converse}} 
\end{figure}

\subsection{Impossibility result for linear codes}\label{linear}

The converse result for linear codes used over a binary erasure channel is based on the 
observation that 
when the update-efficiency is low, the generator matrix $G$ is very sparse, i.e.,
every row of $G$ has very few non-zero entries. Let the random subset $I\in \{1,\dots,n\}$ denote
the coordinates not erased by the binary erasure channel.
Let $G_I$ denote the submatrix of 
$G$ induced by the unerased received symbols, i.e.,  the columns
of $G$ corresponding to $I$. Then, because $G$ is sparse, it is quite likely
that $G_I$ has several all zero rows, and the presence
of such rows implies a large error probability. We formalize
the argument below.
\begin{theorem}\label{thm:converse2}
Consider using some linear  code  of length $n$, dimension $k$ and update-efficiency
$t$, specified by generator matrix $G$
over BEC($p$). Hence, all rows of $G$ have weight at most $t$. Assume that for some $\epsilon >0$,
\begin{equation*} 
t < \frac{\ln \frac{k^2}{2n\ln(1/\epsilon)}}{2\ln \frac1p}.
\end{equation*}
Then, the average probability of error is at least $1/2- \epsilon$.
\end{theorem}
\begin{IEEEproof}
For linear codes over the binary erasure channel, analyzing the probability of error
essentially reduces to analyzing the probability that the matrix $G_I$ induced by the unerased
columns of $G$ has rank $k$ (note that the rank is computed
over $\ff_2$). To show that the rank is likely to be less
than $k$ for sufficiently small $t$, let us first compute the expected number of all zero
rows of $G_I$. Since every
row of $G$ has weight at most $t$, the expected number of all zero rows of $G_I$
is at least
$kp^t$. The rank of $G_I$, $\rank(G_I)$, is at most $k$ minus the number of all zero rows, 
so the expected
rank of $G_I$ is at most
$k-kp^t$.

Now, observe that the rank is a $1$-Lipschitz functional of the independent 
random variables 
denoting the erasures introduced by the channel. 
Therefore, by Azuma's
inequality \cite[Theorem~7.4.2]{AS2000}, the rank of $G_I$ satisfies 
\begin{equation*} 
\Pr(\rank(G_I) \ge \avg \rank(G_I)+\lambda) <
e^{-\frac{\lambda^2}{2n}}.
\end{equation*}
Therefore, 
\begin{equation*}
\Pr(\rank(G_I)  \ge k-kp^t+\lambda) < e^{-\frac{\lambda^2}{2n}}.
\end{equation*}
In particular, substituting $\lambda = kp^t$,
\begin{equation*}
\Pr(\rank(G_I)  = k) < e^{-\frac{k^2p^{2t}}{2n}}.
\end{equation*}
Assuming 
 the value given for $t$, we see that
\begin{equation*}
\Pr(\rank(G_I)  = k) <\epsilon. 
\end{equation*}
Since even the maximum likelihood decoder makes an error with probability
at least $0.5$ when $\rank(G_I) <k$, this shows that when 
\begin{equation*} 
t < \frac{\ln \frac{k^2}{2n\ln(1/\epsilon)}}{2\ln \frac1p},
\end{equation*}
the probability
of error is at least $1/2-\epsilon$. (In fact, the average error probability
converges to 1. The above argument can easily be extended
to show that the probability of decoding successfully
is at most $e^{-\Omega(k^\delta/\log k)}$ for some $\delta >0$, but
we omit the details.)
\end{IEEEproof}

\subsection{Impossibility for general codes}\label{sec:con_general}

Now, we prove that even nonlinear codes cannot have low update-efficiency for the binary symmetric channel.
The argument is based on a simple observation. If a code has dimension $k$ and update-efficiency $t$,
then any given codeword has $k$ neighboring codewords within distance $t$, corresponding to the $k$
possible single-bit changes to the information bits. If $t$ is sufficiently small, it is not possible to 
pack $k+1$ codewords into a Hamming ball of radius $t$ and maintain a low probability of error. 
\begin{theorem}\label{thm:converse3}
Consider using some (possibly non-linear) code of length $n$, size $2^k, k \in \reals_{+}$, and update-efficiency 
$t$ over BSC($p$). Assume that
$t \leq (1-\alpha)\log k/\log((1-p)/p)$, for some $\alpha >0$.
Then, the average probability of error is at least $1-o(1)$, where
$o(1)$ denotes a quantity that goes to zero as $k \rightarrow \infty$.
\end{theorem}
\begin{IEEEproof}
First, we show that a code consisting of $k+1$ codewords contained in a Hamming ball
of radius $t$ has large probability of error. Instead of analyzing BSC($p$),
consider the closely related channel where
exactly $w$ uniformly random errors are introduced, where $w +t  \le n/2$. For this channel, subject to the constraint
that the $k+1$ codewords are contained in a Hamming ball of radius $t$, the
average probability of error is at least
\begin{equation*}
1-\frac{(2t+1){n \choose w+t}}{(k+1){n \choose w}} \geq
1-\frac{2t(n-w)^t}{kw^t}.
\end{equation*}
To see this, take $\bfx_1, \dots,\bfx_{k+1}$ to be the codewords, and $B_i, i =1,\dots, k+1,$
to be the corresponding decoding regions. Without loss of generality, we can 
assume that 1) the ML decoder is detetrministic, so the $B_i$'s are all disjoint, and 2) the codewords 
are all contained in a Hamming ball centered at the zero vector.
Now, let $D_i$ be the set of possible outputs of the channel for input $\bfx_i, i =1, \dots, k+1.$
The average probability of correct decoding is 
\begin{equation*}
 \frac1{k+1}\sum_{i=1}^{k+1}\frac{|B_i\cap D_i|}{\binom{n}{w}}
= \frac1{k+1}\frac{|\cup_i(B_i\cap D_i)|}{\binom{n}{w}} \le  \frac{|\cup_i D_i |}{(k+1)\binom{n}{w}} .
\end{equation*}
But 
$$|\cup_i D_i| \le \sum_{j=-t}^t \binom{n}{w+j} \le (2t+1)\binom{n}{w+t}.$$
The first inequality follows because an erroneous vector can have weight at least $w-t$ and at most $w+t$.
The second inequality follows because $\binom{n}{i}$ increases with $i\le n/2.$ 

Now, with probability at least $1-o(1)$, the number of errors introduced by the  the binary symmetric channel is at least $pn-n^{2/3}$
and at most $pn+n^{2/3}$, 
and the error vectors are uniformly distributed\footnote{This, of course, follows from standard large deviation inequalities such 
as Chernoff bound. For example, see textbook \cite[Thm. 3.1.2]{cover2012elements}.}. 

If $t \leq (1-\alpha)\log k/\log((1-p)/p)$, then for $p <\frac12$, $pn +n^{2/3}+ t < n/2$, for every sufficiently large $n$. And, therefore, the probability of
error on the binary symmetric channel is at least
\begin{equation*}
1-\frac{2t(1-p)^t}{kp^t}+o(1) = 1-
2t/k^\alpha+o(1).
\end{equation*}

Now, for each message $\bfx$ of the given $(n,k)$ code with update-efficiency
 $t$, consider the subcode
$\cC_\bfx$ consisting
of the $k+1$ codewords $\phi(\bfx),\phi(\bfx+\bfe_1),\ldots,\phi(\bfx+\bfe_k)$, corresponding to the 
encodings of $\bfx$ and the $k$ messages
obtained by changing a single bit of $\bfx$. These codewords lie within
a Hamming ball of radius $t$ centered around $\phi(\bfx)$. The 
above argument shows that even a maximum likelihood decoder has a large average probability of error
for decoding the subcode $\cC_\bfx$. Let us call this probability $P_{\cC_\bfx}$. We claim that the average probability of error
of the code $\cC$ with maximum likelihood decoding, $P_{\cC}$, is at least the average, over all
$\bfx$, of the probability
of error for the code $\cC_\bfx$, up to some factor.  In particular,
\begin{equation*} 
P_{\cC} \ge \frac{k}{n}\frac{1}{|\cC|}\sum_{\bfx\in \cC}  P_{\cC_\bfx}.
\end{equation*}
We will now prove this claim and thus the theorem. Note that $P_{\cC}
=  1/|\cC|\sum_{\bfx\in \cC} P_\bfx$, where $P_\bfx$ is the  
probability of error if codeword $\bfx$ is transmitted. Therefore,
\begin{equation*} 
P_{\cC_\bfx} =  \frac{1}{|\cC_\bfx|}\sum_{\bfy\in \cC_\bfx} P_\bfy.
\end{equation*}
We have, 
\begin{align*}
\frac{1}{|\cC|}\sum_{\bfx\in \cC}  P_{\cC_\bfx} & \le \frac{1}{|\cC|}\sum_{\bfx\in \cC} \frac{1}{|\cC_\bfx|}\sum_{\bfy\in \cC_\bfx} P_\bfy\\
&= \frac{1}{(k+1)|\cC|}\sum_{\bfx\in \cC} \sum_{\bfy\in \cC_\bfx} P_\bfy\\
&= \frac{1}{|\cC|} \sum_{\bfx\in \cC} \frac{ d_\bfx}{k+1} P_\bfx,
\end{align*}
where  $d_\bfx =|\{\bfy: \bfx \in \cC_\bfy\}| \le n$. Hence,
\begin{equation*} 
\frac{1}{|\cC|}\sum_{\bfx\in \cC}  P_{\cC_\bfx}  \le \frac{1}{|\cC|} \sum_{\bfx\in \cC} \frac{n}{k+1} P_\bfx = \frac{n}{k+1}P_{\cC}.
\end{equation*}
We conclude that the original code $\cC$ has probability of error
at least $1-o(1)$ when
\begin{equation*} 
t \leq \frac{(1-\alpha)\log k}{\log(\frac{1-p}{p})}.
\end{equation*}
\end{IEEEproof}

\begin{remark} This argument does not work for the binary erasure channel. In fact,
there exist zero rate codes for the binary erasure channel with vanishing error probability
and sub-logarithmic update-efficiency. Specifically, consider an encoding from $k$ bits to $2^k$ bits
that maps a message $\bfx$ to the string consisting of all $0$'s except for a single $1$ in the position 
with binary expansion $\bfx$. Repeat every symbol of this string $c$ times to obtain the final encoding
$\phi(\bfx)$. The update-efficiency
is $2c$, since every codeword has exactly $c$ $1$'s, and different codewords
never have a nonzero entry in the same position. Since the location
of a nonzero symbol uniquely identifies the message, the error probability is at most the 
probability that all $c$ $1$'s in the transmitted codeword are erased, i.e., at most $p^c$. 
Therefore, we achieve vanishing error probability as long as $c\rightarrow
\infty$, and $c$ can grow arbitrarily slowly. 

We conjecture that for positive rates,
even nonlinear codes must have logarithmic
update complexity for the binary erasure channel.
\end{remark}

\subsection{Ensemble of LDGM codes}\label{ensemble}
Let us motivate the study of one particular ensemble of LDGM codes
here.  Suppose we want to construct a code with update-efficiency $t$.
From proposition \ref{prop:linear}, we know that a linear code with
update-efficiency $t$ always has a generator matrix with maximum row
weight $t$.  For simplicity we consider generator matrices with
all rows having weight exactly $t$.  We look at the ensemble of
linear codes with such generator matrices, and show that almost all
codes in this ensemble are bad for $t$ less than certain value.  Note
that any $k \times n$ generator matrix with row weight at most $t$ can
be extended to a generator matrix with block-length $n+t-1$ and row
weight exactly $t$ (by simply padding necessary bits in the last $t-1$
columns).


Let $\Gamma_{n,k,t}$ be the set of all $k \times n$ matrices over $\ff_2$ such that each
 row has exactly $t$ ones. First of all, we claim that almost all the matrices in $\Gamma_{n,k,t}$
 generate codes with dimension $k$ (i.e., the rank of the matrix is $k$). Indeed, we quote the 
 following lemma from \cite{calkin1997dependent}.
\begin{lemma}\label{lem:rank}
Randomly and uniformly choose a matrix $G$ from $\Gamma_{n,k,t}$. 
If $k\le \Big(1-e^{-t}/\ln{2}-o(e^{-t})\Big)n,$ then with probability $1-o(1)$ the rank of $G$
is $k$.
\end{lemma} 
This lemma, along with the next theorem, which is the main result of this section,
will show the fact claimed at the start of this section.

\begin{theorem}\label{thm:converse}
Fix an $0<\alpha<1/2$.  For at least a $1-t^2n^{2\alpha}/(n-t)$
proportion of the matrices in $\Gamma_{n,k,t}, k \ge n^\alpha,$ the
corresponding linear code has probability of error at least $n^\alpha
2^{-\lambda_p t}/\sqrt{t}$ over a BSC($p$), for $p <1/2$ and
$\lambda_p = -1-1/2\log p - 1/2 \log(1-p) >0$.
\end{theorem}
The proof of this theorem is deferred
until later in this section. 
This theorem implies that for any $\alpha <1/2$, most codes in
the random ensemble
of codes with fixed row-weight (and hence update-efficiency)
$t<\alpha/\lambda_p \log n$ have  
probability of error bounded away from $0$ for any positive rate.
Indeed, we have the following corollary.
\begin{corollary}\label{cor:main}
For at least $1-o(1)$ proportion of all linear codes with fixed $t$-row-weight generator matrix,
  $t< (\alpha/\lambda_p) \log n$, $\alpha<\frac12$,  
and dimension $k >n^\alpha$, the probability of error is $1-o(1)$
over a BSC($p$), for $0< p \le 1/2$. 
 \end{corollary}
In particular, this shows that almost all linear codes with fixed row weight $t <
1/(2\lambda_p)\log n$ and rate greater than $1/\sqrt{n}$ are 
 bad (result in high probability of error).
\begin{IEEEproof}[Proof of Corollary \ref{cor:main}]
From Lemma \ref{lem:rank}, it is clear that a $1-o(1)$
proportion of all codes in $\Gamma_{n,k,t}$ have rank $k.$
Hence, if a $1-o(1)$ proportion of codes in $\Gamma_{n,k,t}$ have 
some property, a $1-o(1)$ proportion of codes with  $t$-row-weight generator matrix and
dimension $k$ also have that property.

Now, plugging in the value of $t$ in the expression for probability of
error in Theorem \ref{thm:converse}, we obtain the corollary.
\end{IEEEproof}

To prove Theorem \ref{thm:converse}, we will need the following series of
lemmas.

\begin{lemma}\label{lem:prob1}
Let $\bfx\in \{0,1\}^n$ be a vector of weight $t.$ Let the all-zero
vector of length $n$ be transmitted over a BSC with flip probability
$p<1/2$. If the received vector is $\bfy$, then 
\begin{equation*} 
\Pr(\wt(\bfy) > \dist(\bfx,\bfy)) \ge \frac{1}{\sqrt{t}}2^{-\lambda_p t},
\end{equation*}
where $\lambda_p = -1-1/2\log p - 1/2 \log(1-p) >0.$
\end{lemma}
\begin{IEEEproof}
Let $I \subset [n]$ be the support of $\bfx$. We have $|I| =t.$
Now, $\wt(\bfy) > \dist(\bfx,\bfy)$ whenever the number of errors introduced
by the BSC in the coordinates $I$ is $>t/2.$
Hence,
\begin{align*}
\Pr(\wt(\bfy) > \dist(\bfx,\bfy))  &=\sum_{i>t/2} \binom{t}{i}p^i(1-p)^{t-i}\\
&\hspace{-0.5in}>  \binom{t}{t/2}p^{t/2}(1-p)^{t-t/2}
\ge \frac{1}{\sqrt{t}}2^{-\lambda_p t}.
\end{align*}
\end{IEEEproof}

\begin{lemma}\label{lem:pair}
Suppose two random vectors $\bfx,\bfy \in \{0,1\}^n$ are chosen
independently and uniformly from the set of all length-$n$ binary
vectors of weight $t$. Then,
\begin{equation*} 
\Pr(\supp(\bfx)\cap \supp(\bfy) = \emptyset) > 1-\frac{t^2}{n-t+1}.
\end{equation*}
\end{lemma}
\begin{IEEEproof}
The probability in question equals
\begin{align*}
\frac{\binom{n-t}{t}}{\binom{n}{t}} &= \frac{((n-t)!)^2}{(n-2t)!n!}\\
&= \frac{(n-t)(n-t-1)(n-t-2)\dots(n-2t+1)}{n(n-1)(n-2)\dots(n-t+1)}\\
&= \Big(1-\frac{t}{n}\Big)\Big(1-\frac{t}{n-1}\Big)\dots \Big(1-\frac{t}{n-t+1}\Big)\\
&> \Big(1-\frac{t}{n-t+1}\Big)^t \ge 1-\frac{t^2}{n-t+1}.
\end{align*}
In the last step we have truncated the series expansion of $\Big(1-\frac{t}{n-t+1}\Big)^t$
after the first two terms. The inequality will be justified if the terms of the series
are decreasing in absolute value. Let us verify that to conclude the proof. In the
following $X_i$ denote the $i$th term in the series, $0\le i \le t.$
\begin{align*}
 \frac{X_{i+1}}{X_i}= \frac{\binom{t}{i+1}}{\binom{t}{i}}\cdot\frac{t}{n-t+1}= \frac{t-i}{i+1}\cdot \frac{t}{n-t+1} \le 1,
\end{align*}
for all $i \le t-1.$

\end{IEEEproof}
\begin{lemma}\label{lem:union}
Let us choose any $n^{\alpha}, 0<\alpha<1/2,$ random vectors of weight $t$
independently and uniformly from the set of weight-$t$ vectors. Denote the vectors by
$\bfx_i, 1\le i \le n^\alpha.$ Then, 
\begin{equation*} 
\Pr(\forall i\ne j,  \supp(\bfx_j)\cap \supp(\bfx_i) =\emptyset) \ge  1-\frac{t^2n^{2\alpha}}{n-t}.
\end{equation*}
This implies all of the vectors have disjoint supports with probability
at least $1-t^2n^{2\alpha}/(n-t)$.
\end{lemma}
\begin{IEEEproof}
Form Lemma \ref{lem:pair}, for any pair of randomly and uniformly chosen vectors, 
the probability that they have overlapping support is at most $t^2/(n-t)$.
The claim follows by taking a union bound over all $\binom{n^\alpha}{2}$ pairs of the
randomly chosen vectors.
\end{IEEEproof}

Now, we are ready to prove Theorem \ref{thm:converse}.

\begin{IEEEproof}[Proof of Theorem \ref{thm:converse}]
We begin by choosing a matrix $G$ uniformly at random from $\Gamma_{n,k,t}$.
This is equivalent of choosing each row of $G$ 
uniformly and independently from the set of all $n$-length $t$-weight
binary vectors. Now, $k>n^\alpha$, hence there exists $n^\alpha$ vectors
among the rows of $G$ such that any two of them have disjoint support
with probability at least $1-t^2n^{2\alpha}/(n-t)$ (from Lemma
\ref{lem:union}). Hence, 
for at least a proportion $1-t^2n^{2\alpha}/(n-t)$ of matrices of 
$\Gamma_{n,k,t}$, there are $n^\alpha$ rows with disjoint supports. Suppose
$G$ is one such matrix. It remains to show that the code $\cC$ defined by
$G$ has probability of error at least $n^\alpha 2^{-\lambda_p t}/\sqrt{t}$
over BSC($p$).

Suppose, without loss of generality, that the all zero vector
is transmitted over a BSC($p$), and $\bfy$ is the vector received.
 We know that there exists at least $n^\alpha$
codewords of weight $t$ such that all of them have disjoint support. 
Let $\bfx_i, 1\le i \le n^\alpha$, be those codewords.
Then, the
probability that the maximum likelihood decoder incorrectly decodes $\bfy$
to $\bfx_i$ is 
\begin{equation*} 
\Pr(\wt(\bfy) > \dist(\bfx_i,\bfy)) \ge \frac{1}{\sqrt{t}}2^{-\lambda_p t}
\end{equation*}
from Lemma \ref{lem:prob1}.
As the codewords $\bfx_1, \dots \bfx_{n^\alpha}$ have disjoint supports,
the probability that the maximum likelihood decoder incorrectly decodes to
any one of them is at least
\begin{equation*} 
1- \Big(1-\frac{1}{\sqrt{t}}2^{-\lambda_p t}\Big)^{n^\alpha} = (1-o(1))\cdot\frac{n^\alpha}{\sqrt{t}}2^{-\lambda_p t}.
\end{equation*}
\end{IEEEproof}

\begin{remark}
Theorem~\ref{thm:converse} is also true for the random ensemble of matrices 
where the entries are independently chosen from $\ff_2$ with $\Pr(1) =t/n$.
\end{remark}
%

\section{Impossibility result for local recovery}\label{sec:recover}
Now that we have our main impossibility result on update-efficiency we turn to
the local recovery property. In this section, we
  deduce  the converse result concerning local recovery for the binary erasure channel.
We show that any code with a given   local recoverability has to have
rate bounded away from capacity to provide arbitrarily small probability of error,
when used over the binary
erasure channel. 
In particular,  for any code, including non-linear codes, recovery
complexity at a gap of $\epsilon$ to capacity on the BEC must be at least
$\Omega(\log 1/\epsilon)$, proving that the above LDPC construction is simultaneously
optimal to within constant factors for both update-efficiency and local recovery.


The intuition for the converse is that if a code has low
local recovery complexity, then  
codeword positions can be predicted by looking at a few
codeword symbols. As we will see, this implies that the code
rate must be bounded away from capacity, or the probability of
error approaches $1$. In a little more detail, for an 
erasure channel, the average error probability is related to 
how the codewords behave under projection onto the unerased
received symbols. Generally, different codewords
may result in the same string under projection, and
without loss of generality, the ML decoder can be assumed
to choose a codeword from the set of codewords matching the
received channel output in the projected coordinates uniformly
at random. Thus, given a particular erasure pattern induced
by the channel, the average probability of decoding success 
for the ML decoder is simply the number of different
codeword projections, divided by $2^{Rn}$, the size of the 
codebook. We now show that the number of different projections
is likely to be far less than $2^{Rn}$.

\begin{theorem}\label{thm:local_converse}
Let $\cC$ be a code of  length $n$ and rate $1-p -\epsilon$, $\epsilon>0$, that
achieves probability of error strictly less than $1$, when used over
BEC($p$). Then, the local recoverability of $\cC$ is at least $c\log
1/\epsilon$, for some constant $c>0$ and $n$ sufficiently large.
\end{theorem}
\begin{proof}
Let $\cC$ be a code of length $n$ and size $2^{nR}$ that has local
recoverability $r$.  Let $T$ be the set of coordinates with the
property that the query positions required to recover these
coordinates appear before them. To show that such an ordering exists
with $|T| \ge n/(r+1)$, we can randomly and uniformly permute the
coordinates of $\cC$. The expected number of such coordinates is then
$n/(r+1)$, hence some ordering exists with $|T| \ge n/(r+1)$.


Assume  $I \subseteq \{1,\dots, n\}$ is the set
of coordinates erased by the BEC, and let $\bar{I} = \{1,\dots, n\} \setminus I.$
Let $\bfx\in \cC$ be a randomly and uniformly chosen codeword. $\bfx_I$ and
$\bfx_{\bar{I}}$ denote the projection of $\bfx$ on the respective coordinates.
We are interested in the logarithm of the number of different
codeword projections onto $\bar{I}$, which we denote 
by $\log{S(\bfx_{\bar{I}})}$. Note that this is a 
random-variable with respect to the random choice of $I$ by the BEC.

Suppose that the number of elements of $T$ that have all $r$ of their recovery positions 
un-erased is $u$. Then, the number of different codeword projections is unchanged
if we remove these $u$ elements from $T$. Hence,
\begin{equation*} 
\log{S(\bfx_{\bar{I}})} \le |\bar{I}| - u.
\end{equation*}
But $\avg u \ge (1-p)^r |T|$.  Therefore,
\begin{equation*}
\avg \log{S(\bfx_{\bar{I}})} \le n(1-p) - (1-p)^r \frac{n}{r+1}.
\end{equation*}

Observe that $\log{S(\bfx_{\bar{I}})}$ is a $1$-Lipschitz functional of
independent random variables (erasures introduced by the channel). This 
is because projecting onto one more position cannot decrease the number
of different codeword projections, and at most doubles the number
of projections. Therefore, we can use Azuma's inequality to conclude
that
\begin{equation*}
\Pr\Big(\log{S(\bfx_{\bar{I}})} > n(1-p) - (1-p)^r \frac{n}{r+1} + \epsilon n  \Big) \le e^{-\frac{\epsilon^2 n}2}.
\end{equation*}
If we have,
$$
r \le \frac{\log\frac1{3\epsilon}}{\log\frac2{1-p}},
$$
then,
$$
\frac{(1-p)^r}{1+r} \ge \Big(\frac{1-p}2\Big)^r \ge 3\epsilon.
$$
%
%
%
But this implies,
\begin{equation*}
\Pr\Big(\log{S(\bfx_{\bar{I}})} > n(1-p-2\epsilon)  \Big) \le e^{-\frac{\epsilon^2 n}2}.
\end{equation*}
This means that for a suitable constant $c$, if $r \le c\log1/\epsilon$, then
with very high probability
$\log{S(\bfx_{\bar{I}})} \le n(1-p-2\epsilon)$.
However, there are $2^{Rn}=2^{n(1-p-\epsilon)}$ codewords, so we conclude
that the probability 
of successful decoding is at most 
\begin{equation*}
2^{-\epsilon n}+e^{-\frac{\epsilon^2 n}2}.
\end{equation*}
Thus, we have proved that if $r \le c\log1/\epsilon$,
the probability of error converges to $1$, and in particular, is larger than
any $\alpha<1,$ for sufficiently large $n$.
\end{proof}

\begin{remark}
Rather than considering the number of different codeword projections,  we could have considered
the entropy of the distribution of codeword projections onto $I$, which is also a $1$-Lipschitz
functional. This is a more general approach that can be extended to the case where local recovery can 
be adaptive and randomized, and only has to succeed with a certain probability (larger than $.5$), as opposed to 
providing guaranteed recovery. However, one obtains a bound of $n(1-p-2\epsilon)$
on the the entropy, so Fano's inequality only shows that the probability of error 
must be $\Omega(\epsilon)$, while the above analysis shows that the probability of error
must be close to $1$. 
%
\end{remark}

\section{General update-efficient codes} \label{sec:general}
In this section we discuss some further observations regarding the update-efficiency
of codes. Let us now give a more general definition of update-efficiency that
we started with in the introduction.
\begin{definition}
A code is called $(u,t)$-update-efficient if, for any $u$ bit changes in the message,
the codeword changes by at most $t$ bits. In other words, the code $(\cC,\phi)$ is  {\em $(u,t)$-update-efficient}
if
for all $\bfx \in \ff^k_2,$ and for
all $\bfe\in \ff_2^k : \wt(\bfe) \le u,$ we have $\phi(\bfx +\bfe) = \phi(\bfx) + \bfe',$
for some $\bfe' \in \ff_2^n: \wt(\bfe') \le t.$
\end{definition}

It is easy to see that an $(1,t)$-update-efficient code is a code with update-efficiency $t.$
As discussed earlier, any $(u,t)$-update-efficient code must satisfy $t > d$, the minimum distance of
the code. In fact, we can make a stronger statement.

\begin{proposition}
Suppose a $(u,t)$-update-efficient code of length $n$, dimension $k$, and minimum distance $d$
exists. Then, 
$$
\sum_{i=0}^u \binom{k}{i} \le B(n,d,t),
$$
where $B(n,d,w)$ is the size of the largest code with distance $d$ such that
each codeword has weight at most $w.$
\end{proposition}
\begin{IEEEproof}
Suppose $\cC$ is an update-efficient code, where
$\bfx \in \ff_2^k$ is mapped to $\bfy \in \ff_2^n.$ The
$\sum_{i=0}^u \binom{k}{i}$ different message vectors
within distance $u$ from $\bfx$ should map to  codewords within 
distance $t$ from $\bfy.$ Suppose these codewords are
$\bfy_1, \bfy_2, \dots.$ Consider the vectors $\bfy-\bfy,
\bfy_1-\bfy, \bfy_2-\bfy, \dots.$ These must be at least distance $d$ apart
from one another and all of their weights are at most $t$. This proves
the claim.
\end{IEEEproof}

There are a number of useful upper bounds on the maximum size of
constant weight codes (i.e., when the codewords have a constant weight
$t$) that can be used to upper bound $B(n,d,t).$ Perhaps the most
well-known bound is the Johnson bound \cite{J1962}. An easy extension
of this bound says $B(n,d,t) \le dn /(dn -2tn+2t^2)$, as long as
the denominator is positive. However, this bound is not very
interesting in our case, where we have $n \gg t \ge d.$ The
implications of some other bounds on $B(n,d,t)$ on the parameters of
update-efficiency is a topic of independent interest.

Note that any code with update-efficiency $t$ is a $(u,ut)$-update-efficient code. Hence, from Section \ref{explicit}, we can 
 construct an $(u,O(u\log n))$ update-efficient code that achieves 
the capacity of a BSC($p$).
On the other hand one expects a converse result of the form
$$
\sum_{i=0}^u \binom{k}{i} \le K(n,t,p),
$$
where $K(n,t,p)$ is the maximum size of a code with codewords having weight bounded 
by $t$ that achieves arbitrarily small probability of error. 
Indeed, just by emulating the proof of Theorem \ref{thm:converse3}, we obtain
the following result.
\begin{theorem}\label{thm:converse3gen}
Consider using some (possibly non-linear) $(u,t)$-update-efficient code of length $n$, and 
dimension (possibly fractional) $k$  over BSC($p$). Assume that
\begin{equation*} 
t \leq \frac{(1-\alpha)\log \sum_{i=0}^u
  \binom{k}{i}}{\log(1-p)/p},
\end{equation*}
for any $\alpha >0$. 
Then, the average probability of error is at least $1-o(1)$, where
$o(1)$ denotes a quantity that goes to zero as $k \rightarrow \infty$.
\end{theorem}

This shows that the $(u,O(u\log n))$ update-efficient code constructed
by the method of Section \ref{explicit}, is almost optimal for  $u \ll n.$

\begin{remark}[Bit error rate and error reduction codes]
Suppose we change the model of update-efficient code in the following way (limited to 
only this remark). The encoding $\phi: \ff_2^k \to \ff_2^n$ and decoding $\theta: \ff_2^n \to \ff_2^k,$
is such that for a random error vector $\bfe$ induced by the BSC($p$) and any $\bfx \in \ff_2^n$, 
$\dist(\theta(\phi(\bfx)+\bfe), \bfx ) \sim o(k)$ with high probability. This can be thought of as an error-reducing code or 
a code with low message bit error rate \cite{kochman2012adversarial}. Under this notion, error-reducing 
codes are update-efficient.
When the message changes $\le u$ bits from the previous state $\bfx\in \ff_2^k$, 
we do not change the codeword. Then, the decoder output will be within $o(k)+u$ bits
from the original
message.
\end{remark}


\section{Rate-Distortion Counterparts}
\label{ratedistortion}

In this paper, we have focused on error correcting codes possessing
good update efficiency and local recovery properties. In principle,
these properties are also applicable to the problem of lossy source coding.
Informally, lossy source coding is concerned with optimally compressing a source
so that the source can be reconstructed up to a specified distortion from its
compressed representation.
We refer the reader to any standard textbook on information theory (e.g., \cite{cover2012elements})
for a formal definition of the lossy source coding problem in terms of a source
and distortion model.
The associated rate-distortion function $R(D)$ expresses the
optimal (smallest) rate achievable given a tolerable reconstruction distortion $D$. 


Update-efficiency and local recoverability have natural analogs for lossy
source codes. In more detail, update-efficiency can be measured by asking
how much the encoding (compression) of the source changes when the source is changed
slightly, e.g., how many bits of the compression change when a single bit of the original
source is changed. In the context of lossless source coding, update-efficient
codes have been considered before in several papers, e.g., \cite{montanari2008smooth}. 
Local recoverability for a lossy source code can be measured
by the number of bits of the compression that must be queried in order to recover
any particular symbol of the source reconstruction. That is, a lossy source code has good
local recoverability if, for all indices $i$, few bits of the compression must be read in order to compute
the $i^{th}$ symbol of the lossy reconstruction.

The main questions to be asked, in the spirit of this paper, are 1) if we allow a compression rate slightly
above the optimal rate specified by rate-distortion function, i.e., 
rate $R(D)+\epsilon$, what is the best possible local
recoverability, and 2) what is the best possible update-efficiency (in terms of $\epsilon$)?
As a simple example, we consider these questions in the context of compressing a uniform
binary source under Hamming distortion (again, the reader is referred to standard 
information theory textbooks such as \cite{cover2012elements}
for a formal definition of this model). In this case, we briefly describe known results that allow
us to make some progress on questions 1 and 2. First, it can be shown that local recoverability 
must grow as
$\Omega(\log(1/\epsilon))$, that is, at least $\Omega(\log(1/\epsilon))$ bits of the compression
must be queried to recover the $i^{th}$ symbol of the reconstruction. This is a corollary of 
known results for
LDGM codes (Theorem 5.4.1 from \cite{chandar2010sparse}), and the proof 
given there already applies to arbitrary
(even non-linear) codes. It is known 
that LDGM codes can achieve $O(\log(1/\epsilon))$ local recovery complexity,
so in this simple case, the local recoverability can be characterized up to a constant factor.

Update-efficiency, on the other hand, remains an open question, even for this simple model.
We note that update-efficiency of $O(1/\epsilon\log(1/\epsilon))$ can be achieved
via random codes. This can be seen by the following argument.
First, it is easily verified that a random code of length
$O(1/\epsilon\log(1/\epsilon))$ 
can achieve rates at most $\epsilon$ above the rate-distortion function. 
Copying the strategy from Section~\ref{explicit}, we construct a length $n$ 
lossy source code by splitting the $n$-bit source into blocks of length 
$O(1/\epsilon\log(1/\epsilon))$, and apply the
random code constructed above to each block. Changing one bit of the source only affects
the corresponding block, so changing one bit of the source
requires updating at most $O(1/\epsilon\log(1/\epsilon))$ bits of the compression.
The length $n$ code clearly has the same compression rate as the base random code,
and achieves exactly the the same expected distortion because of the linearity of expectation and the memoryless
nature of the source and distortion models considered.
Therefore, $O(1/\epsilon\log(1/\epsilon))$ update-efficiency is achievable.
However, it is unclear that this is optimal. In particular, 
we are unaware of any lower bound showing that the update-efficiency has to scale with $\epsilon$ at all.
A thorough study of local
recoverability and update-efficiency for lossy source coding is left for future work.

\section{Concluding Remarks}
\label{sec:conc}

Although our results are derived for binary-input channels, as opposed to the
large alphabet channel models usually considered for distributed storage, our proofs
extend  to large alphabet case. The $q$-ary generalizations for BSC and BEC are respectively the
{\em $q$-ary symmetric channel} and {\em $q$-ary erasure channel}. The definitions and capacities of these
channels are standard and can be found in textbooks, for example, in \cite[\S1.2 \& \S1.5.3]{roth2006introduction}. 

The existential result of Section~\ref{explicit} extends to the case of $q$-ary channels. See, \cite[\S IIIB, remark 6]{barg2002random} 
for more detail on how the error-exponent result for BSC extends to $q$-ary case. There are also results
concerning the error-exponent of $q$-ary {\em low density} codes that can be used to extend Theorem
\ref{thm:exists}. The result one can most directly use is perhaps \cite{hof2009performance}.

The converse results for Section~\ref{sec:LDGM} and \ref{sec:recover}, in particular Theorem~ \ref{thm:converse2}, Theorem~ \ref{thm:converse3} and
Theorem~ \ref{thm:local_converse} can easily be stated for the case of $q$-ary channel. The observations
regarding adversarial error case of Section~\ref{adversarial} is also extendable to $q$-ary case in a straight-forward
manner.

\section*{Acknowledgment}

A.~M.~thanks Alexander Barg, Yury Polyanskiy  and Barna Saha for useful
discussions.

\bibliographystyle{abbrv}
\bibliography{aryabib}

\end{document}